\documentclass[12pt]{article}
\usepackage{geometry}
\usepackage{amsthm, amsmath, amsmath, amssymb, bbm, bm}
\usepackage[round]{natbib}
\usepackage[colorlinks, 
linkcolor=blue,
anchorcolor=blue,
urlcolor=blue,
citecolor=blue]{hyperref}
\usepackage{graphicx}
\usepackage{mathrsfs}
\usepackage{threeparttable}
\usepackage[title]{appendix}
\usepackage{url}
\usepackage{diagbox}
\usepackage{booktabs}
\usepackage{caption}
\usepackage{subfigure}
\usepackage{float}
\usepackage{listings}
\usepackage{multirow}
\usepackage{rotating}
\usepackage{longtable}
\usepackage{enumitem}
\usepackage{cases}
\usepackage{algorithm}
\usepackage{algorithmic}
\usepackage{ulem}
\usepackage{authblk}

\usepackage{xr}
\newtheorem{theorem}{Theorem}[section]

\newtheorem{proposition}{Proposition}[section]

\newtheorem{lemma}{Lemma}[section]
\newtheorem{definition}{Definition}[section]

\newtheorem{assumption}{Assumption}

\newtheoremstyle{mycase}{5pt}{5pt}{\upshape}{}{\bfseries}{.}{ }{}\theoremstyle{mycase}

\newtheorem{remark}{Remark}
\usepackage{xcolor}  
\newtheoremstyle{myIV}{1.5pt}{1.5pt}{\upshape}{1em}{\bfseries}{.}{ }{}\theoremstyle{myIV}

\newtheoremstyle{mysets}{1.5pt}{1.5pt}{\upshape}{}{\bfseries}{.}{ }{} \theoremstyle{mysets}

\renewcommand{\theequation}{\thesection.\arabic{equation}}
\numberwithin{equation}{section}

\renewcommand{\hat}{\widehat}
\renewcommand{\tilde}{\widetilde}

\newcommand{\Cov}{\mathrm{Cov}}
\newcommand{\Var}{\mathrm{Var}}
\newcommand{\bX}{\boldsymbol{X}}

\newcommand{\bbeta}{\bm{\beta}}

\newcommand{\tr}{\mathrm{tr}}
\def\bfred#1{{\color{red}\bf#1}}

\def\tangcomment#1{\vskip 2mm\boxit{\vskip 2mm{\color{red}\bf#1} {\color{blue}\bf --
Yanlin\vskip 2mm}}\vskip 2mm}
\def\shicomment#1{\vskip 2mm\boxit{\vskip 2mm{\color{blue}\bf#1} {\color{blue}\bf --
Yanmei\vskip 2mm}}\vskip 2mm}

\newcommand{\hao}{\textcolor{magenta}}

\begin{document}
	\allowdisplaybreaks[3] 
	
	\title{\bf Estimation and inference of high-dimensional partially
linear  regression models with latent factors }
	
	\author{Yanmei Shi$^{1}$, Meiling Hao$^{2}$, Yanlin Tang $^{3*}$ and Xu Guo$^1$\footnote{Corresponding authors: Xu Guo and Yanlin Tang; email-address: xustat12@bnu.edu.cn and yltang@fem.ecnu.edu.cn.}
		\\~
		\\	{\small \it $^{1}$ School of Statistics, Beijing Normal University, Beijing, China}\\
		{\small \it $^{2}$ School of Statistics, University of International Business and Economics, Beijing, China}\\
		
		{\small \it $^{3}$ Key Laboratory of Advanced Theory and
			Application in Statistics and Data
			Science–MOE, School of Statistics, East
			China Normal University, Shanghai,
			China}\\
	}
	
	\renewcommand\Authands{, and }

	\date{}
	\maketitle
	
	\vspace{-0.5in}
	
	\begin{abstract}

In this paper, we introduce a novel high-dimensional Factor-Adjusted sparse Partially Linear regression Model (FAPLM), to integrate the linear effects of high-dimensional latent factors with the nonparametric effects of low-dimensional covariates. The proposed FAPLM  combines the interpretability of linear models, the flexibility of nonparametric models, with the ability to effectively capture the dependencies among high-dimensional covariates. 
{We develop a penalized estimation approach for the model by leveraging B-spline approximations and factor analysis techniques. Theoretical results establish error bounds for the estimators, aligning with the minimax rates of standard Lasso problems.}
 To assess the significance of the linear component, we  introduce a factor-adjusted projection debiased procedure and employ the Gaussian multiplier bootstrap method to  derive critical values. Theoretical guarantees are provided under regularity conditions.
 Comprehensive numerical experiments validate the finite-sample performance of the proposed method. Its successful application to a birth weight dataset, the motivating example for this study, highlights both its effectiveness and practical relevance.

		\medskip
		
	\end{abstract}

	\section{Introduction} \label{Introduction}

In the era of information explosion,  massive datasets have been accumulated across a wide range of scientific disciplines, including genomics, social sciences, and econometrics \citep{peng2010regularized, belloni2012sparse, fan2014challenges, buhlmann2014high}.
In this context, we  have access to a birth weight dataset, which   consists of  24,539 gene expression profiles from peripheral blood samples of 72 pregnant women, along with clinical data such as maternal age, Body Mass Index (BMI), gestational age, parity, blood cotinine levels, and infant birth weights. 
The study aims to explore the relationship between maternal gene expression and infant birth weight. Preliminary analysis suggests a potential nonlinear association between maternal age and birth weight \citep{sherwood2016partially}. Additionally, the gene expression profiles exhibit significant correlations, likely stemming from shared biological pathways, co-regulatory mechanisms, and technical factors associated with experimental and preprocessing procedures.
Effectively analyzing this dataset requires addressing several key challenges: (1) capturing partially linear relationships between the response variable and covariates, (2) handling high-dimensional data, and (3) accurately modeling dependencies among features. Existing models, however, fall short of fully addressing these requirements. This underscores the pressing need for a novel methodological framework that meets these objectives while providing interpretable and actionable insights. Developing such a model remains both a critical and challenging task.

In this paper, we propose a novel  high-dimensional Factor-Adjusted sparse Partially Linear regression Model (FAPLM), 
\begin{align} Y&=\boldsymbol{u}^{\top}\boldsymbol{\beta}_0+\boldsymbol{f}^{\top}\boldsymbol{\varphi}_0+g_{0}(\boldsymbol{z})+\varepsilon, \label{FAPLM} \\
		\text{with} \ \ \boldsymbol{x}&=\boldsymbol{B}\boldsymbol{f}+\boldsymbol{u}, \label{factor model}
	\end{align}
where  $Y$  is a response variable, $\boldsymbol{x}$ and $\boldsymbol{z}$ are  $p$-dimensional and   $d$-dimensional  predictors, 
 $\boldsymbol{f}$ is a $K$-dimensional vector of latent factors, $\boldsymbol{B}\in \mathbb{R}^{p\times K}$ is the corresponding factor loading matrix, and $\boldsymbol{u}$ is a $p$-dimensional vector of idiosyncratic component which is uncorrelated with $\boldsymbol{f}$. All the three parts $\boldsymbol{B}$, $\boldsymbol{f}$ and $\boldsymbol{u}$ are unobserved. 
The parameter $\boldsymbol{\beta}_0\in \mathbb{R}^{p}$ is a vector of regression parameters, the regression coefficient on $\boldsymbol{f}$ is $\boldsymbol{\varphi}_0=\boldsymbol{B}^{\top}\boldsymbol{\beta}_0$, and   $g_{0}(\cdot)$  is an  unknown link function.
The random error  $\varepsilon$  satisfies that $\mathbb{E}\left(\varepsilon\right)=0$, $\mathbb{\Var}(\varepsilon)={\sigma}_{\varepsilon}^2$, and it's independent of $\boldsymbol{u}$, $\boldsymbol{f}$ and $\boldsymbol{z}$.
 In addition,  we assume that $\boldsymbol{f}$ and $\boldsymbol{z}$ are uncorrelated.
  The proposed model is compelling because it melds the chief advantages of linear regression
 models,  nonparametric  models and factor models. Specifically,  Model \eqref{FAPLM}  integrates a linear structure for  $\boldsymbol{x}$  with a nonparametric structure for $\boldsymbol{z}$, which strikes a balance between the interpretability of linear regression models and the flexibility of nonparametric models. 
{This paper focuses on scenarios where $p$ is high-dimensional, including cases where it diverges with $n$ and $p \gg n$ (where $n$ represents the sample size), while $d$ remains fixed. Without loss of generality, we assume $d = 1$.  }
 In such cases, the components of the observed $p$-dimensional covariate vector $\boldsymbol{x}$ are correlated through a set of shared latent factors, effectively addressing the significant dependencies among features in high-dimensional data.  In summary, the FAPLM incorporates all the
 aforementioned desired features.

The proposed model bridges two important frameworks: factor models and partially linear regression models (PLMs).
The factor model \eqref{factor model} is a widely used tool for capturing dependencies among features \citep{Bai2003, fan2013large}, with extensive research dedicated to its estimation and statistical inference; see e.g., \cite{kneip2011factor, fan2013large,  fan2020factor, fan2021augmented,guo2022doubly,beyhum2023tuning,du2023simultaneous,ouyang2023high,sun2023decorrelating,bing2024inference,fanJ2023}.
 However, these studies often rely on specific model assumptions and face challenges in high-dimensional settings or when dealing with semiparametric structures like PLMs.
Similarly, extensive research has been conducted on PLMs, particularly in high-dimensional settings. Various methods have been developed to address the parameter estimation problem, as exemplified in studies such as \cite{xie2009scad, zhu2017nonasymptotic, lian2019projected}. Furthermore, significant progress has been made in statistical inference for PLMs, with notable contributions from works such as \cite{zhu2019high, wang2017generalized, wang2020test, liu2020tests, zhao2023new}.
 Nevertheless,  
 {many of these methods face challenges in consistently recovering the true model or accurately testing regression coefficients when significant correlations exist among the covariates, thereby limiting their effectiveness in such scenarios.}

Hence, estimating and inferring the unknown parameters in the FAPLM pose significant challenges, particularly when high-dimensional variables exhibit significant dependence. Addressing these challenges requires the development of novel methodologies. The primary contributions of this work are summarized as follows.
 \begin{itemize}
 \item[$\bullet$] 
We propose a general framework for high-dimensional PLMs with latent factors. The primary challenges in Model \eqref{FAPLM} arise from the unknown smooth function $g_0(\cdot)$ and the unobserved factors  $\boldsymbol{f}$ and $\boldsymbol{u}$.  
To address these issues, we employ the B-spline method to approximate the nonparametric component of the FAPLM \citep{huang2007efficient, liang2009variable, xie2009scad}, and use principal component analysis to estimate the latent factors. It is important to note that  estimating $\boldsymbol{f}$ and $\boldsymbol{u}$ introduces additional variability, bringing
further complexity to the problem.
This approach reformulates the original model into a factor-adjusted linear structure,  facilitating more efficient estimation and inference.
 \item[$\bullet$] 
We establish an upper bound for the error of the linear estimator, aligning with the minimax rate for the standard Lasso problem \citep{raskutti2011minimax}. In the context of high-dimensional FAPLM, we demonstrate that the approximation error of the nonlinear component and the estimation error of the latent factors have a negligible impact on the overall error of the linear component. These findings provide valuable insights into the role of latent factors in mediating the interaction between linear and nonlinear components.
 \item[$\bullet$] 
We develop a test procedure based on the  projection strategy to evaluate the significance of the linear component. This task is complicated by the inherent bias of penalized estimators, as well as the challenges posed by the unknown nonparametric component $g_0(\cdot)$ and latent factors, which hinder achieving a zero-centered limiting distribution. To overcome these difficulties, we propose a factor-adjusted projection debiased test, where the test statistic is approximated by the maximum of a high-dimensional zero-mean Gaussian vector. Critical values are derived using the Gaussian multiplier bootstrap method. Theoretical validity of the proposed procedure is established under regularity conditions, and its power properties are comprehensively analyzed.
 \end{itemize}

	The rest of the paper is organized as follows. In Section \ref{FAPLM section}, we  approximate the nonparametric function with polynomial splines, and thoroughly investigate the estimation of latent factors.
	In addition, we consider the regularization estimation of $\bbeta_0$. 
	In Section \ref{Statistical analysis}, we develop a factor-adjusted projection debiased     procedure  for testing whether $\bbeta_0=\boldsymbol{0}$. 
    Section \ref{Theoretical Results} derives the
theoretical properties of the regularized  estimator and factor-adjusted projection debiased estimator.    
    We present the findings of our simulation studies in Section \ref{Numerical studies section} and provide an analysis of real data in Section \ref{Real data analysis section} to assess the performance and effectiveness of the proposed approach. Conclusions and discussions are presented in Section \ref{Conclusions and discussions section}. Proofs of the main Theorems and  related technical Lemmas are attached in the  Supplementary Material.

	\textbf{Notations.} Let $\mathbb{I}(\cdot)$  denote the indicator function. For a vector $\boldsymbol{a}=\left(a_{1}, \ldots, a_{m}\right)^{\top}\in \mathbb{R}^{m}$, we denote its $\ell_{q}$ norm as $\|\boldsymbol{a}\|_{q}=\left(\sum_{\ell=1}^{m}|a_{\ell}|^{q}\right)^{1/q}, \ 1\leq q< \infty$, $\|\boldsymbol{a}\|_{\infty}=\max_{1\leq \ell \leq m}|a_{\ell}|$,  and $\|\boldsymbol{a}\|_{0}=\sum_{\ell=1}^{m}\mathbb{I}(a_{\ell}\neq 0)$.
	For any integer $m$, we define $[m]=\{1,\ldots, m\}$.
	The Orlicz norm  of a scalar random variable $X$ is defined as $\|X\|_{\psi_{2}}=\inf{\{c>0: \mathbb{E}\exp(X^{2}/c^{2})\leq 2\}}$.
	For a random vector $\boldsymbol{x}\in \mathbb{R}^{m}$, we define its  Orlicz norm  as $\|\boldsymbol{x}\|_{\psi_{2}}=\sup_{\|\boldsymbol{c}\|_{2}=1}\|\boldsymbol{c}^{\top}\boldsymbol{x}\|_{\psi_{2}}$.
	Furthermore, we use $\boldsymbol{I}_{K}$,  $\boldsymbol{1}_{K}$ and $\boldsymbol{0}_{K}$ to denote the identity matrix in $\mathbb{R}^{K \times K}$, a vector of dimensional $K$ with all elements being $1$ and all elements being $0$, respectively.
	For a matrix $\boldsymbol{A}={(A_{jk})}$, 
	{ $\boldsymbol{a}_{i}$ represents {the} $i$-th row of $\boldsymbol{A}$, and $\boldsymbol{A}_{j}$ represents {the} $j$-th column of $\boldsymbol{A}$.}
    We define $\|\boldsymbol{A}\|_{\mathbb{F}}=\sqrt{\sum_{jk}{A}_{jk}^{2}}$, $\|\boldsymbol{A}\|_{\max}=\max_{j,k}|A_{jk}|$, $\|\boldsymbol{A}\|_{\infty}=\max_{j}\sum_{k}|A_{jk}|$ and $\|\boldsymbol{A}\|_{1}=\max_{k}\sum_{j}|A_{jk}|$
	to be its Frobenius norm, element-wise max-norm, matrix $\ell_{\infty}$-norm and matrix $\ell_{1}$-norm, respectively.
	Besides, we use $\lambda_{\min}(\boldsymbol{A})$ and $\lambda_{\max}(\boldsymbol{A})$  to denote the minimal and maximal eigenvalues of $\boldsymbol{A}$, respectively.
	We use $|\mathcal{A}|$ to denote the cardinality of a set  $\mathcal{A}$.
	For two positive sequences $\{a_{n}\}_{n \geq 1}$, $\{b_{n}\}_{n \geq 1}$, we write $a_{n}=O(b_{n})$ if there exists a positive constant $C$ such that $a_{n}\leq C\cdot b_{n}$, and we write $a_{n}=o(b_{n})$ if $a_{n}/b_{n}\rightarrow 0$. 
	Furthermore, if $a_{n}=O(b_{n})$ is satisfied, we write $a_n\lesssim b_n$. If
	$a_n\lesssim b_n$ and $b_n\lesssim a_n$, we  write it as $a_n\asymp b_n$ for short.
	In addition, let $a_{n}=O_{\mathbb{P}}(b_{n})$ denote $\Pr(|a_n/b_n| \leq c)\rightarrow 1$ for some constant $c < \infty$. Let $a_{n}=o_{\mathbb{P}}(b_{n})$ denote $\Pr(|a_n/b_n| > c)\rightarrow 0$ for any constant $c >0$.
 The parameters $c, c_{0},  C, C_{1}, C_{2}$ and $K'$  appearing in
	this paper are all positive constants. 

	\section{Methodology} \label{FAPLM section}		
	In this section, we investigate the reformulation of FAPLM and the consistent estimation of latent factors. Based on the transformed model and the estimated factors, we then propose the regularized estimators for the model parameters to ensure accurate inference.
    
	\subsection{Reformulation of FAPLM}
To address the unknown nonparametric component in model \eqref{FAPLM}, several methods have been proposed, including B-splines \citep{sherwood2016partially}, Reproducing Kernel Hilbert Space (RKHS) techniques \citep{wang2022sparse}, and neural networks \citep{zhong2024neural}. In this paper, we adopt B-spline method to approximate the nonparametric component $g_{0}(\cdot)$ in model \eqref{FAPLM}. 
To proceed, we first introduce the relevant definition.

	\begin{definition} (H$\ddot{\text{o}}$lder Space).
    Let  $\mathcal{H}_{r}$ denote the H$\ddot{\text{o}}$lder space of order $r$.
		Define  $\mathcal{H}_{r}$ as the collection of all functions $g: [0,1] \rightarrow \mathbb{R}$, whose $m$-th  order derivative satisfies the H$\ddot{\text{o}}$lder condition of order $\nu$ with $r:= m+\nu$, where $m$ is a positive integer and  $\nu \in (0,1]$.  Specifically, for any $g \in \mathcal{H}_{r}$, there exists a constant $c>0$ such that for any $0 \leq s,t \leq 1$, 
        $\left|g^{(m)}(s)-g^{(m)}(t)\right|\leq c|s-t|^{\nu}$.
	\end{definition}
    
	The order $r$ reflects the smoothness of functions within the  H$\ddot{\text{o}}$lder space and determines the approximation accuracy of the B-spline method for $g_0(\cdot)$.
	Let $\boldsymbol{\pi}_i=\boldsymbol{\Pi}( Z_{i})=\{\Pi_{1}(Z_{i}), \ldots,\Pi_{k_{n}+h+1}(Z_{i})\}^{\top}, \ i=1,\ldots,n$,
	be a set of B-spline basis functions of order $h+1$ with $k_{n}$ quasi-uniform internal knots on $[0, 1]$. 
	We assume that the B-spline basis is non-negative and  normalized to have $\sum_{s=1}^{M_{n}}\Pi_{s}(Z)=\sqrt{M_{n}}$, with $M_{n}=k_{n}+h+1$.
	Then if $g_{0} \in \mathcal{H}_{r}$ for some $r \geq 1.5$, 
	$g_{0}$ can be well approximated by a linear combination of normalized B-spline basis functions. That is, there exists a $M_{n}$-vector $\boldsymbol{\xi}_{0}$ such that 
	\begin{align} \label{unparameter error}
		\sup_{Z \in [0,1]}\left|\boldsymbol{\Pi}( Z)^{\top}\boldsymbol{\xi}_{0}-g_{0}(Z)\right|< CM_{n}^{-r}.
	\end{align} 
	Here $C$ denotes a positive constant that depends solely on the orders of the H$\ddot{\text{o}}$lder space and the spline function space.  This result follows from the approximation properties of splines \citep{huang2003local, schumaker2007spline}.
	With the B-spline basis functions, model \eqref{FAPLM} can be approximated by
	\begin{align} Y&\approx \boldsymbol{f}^{\top}\boldsymbol{\varphi}_0+\boldsymbol{u}^{\top}\boldsymbol{\beta}_0+\boldsymbol{\Pi}( Z)^{\top}\boldsymbol{\xi}_0+\varepsilon, \label{FAPLM_basis_approximation}
		 \\
        \text{with} \ \ \boldsymbol{x}&=\boldsymbol{B}\boldsymbol{f}+\boldsymbol{u}, \notag
	\end{align}
	where $\boldsymbol{\varphi}_0=\boldsymbol{B}^{\top}\boldsymbol{\beta}_0$. 
This demonstrates that FAPLM has been reformulated into a linear form via the B-spline method. 
 Consequently, we can apply linear model-based estimation and inference methods to analyze the parameters in model \eqref{FAPLM_basis_approximation}. However, since the latent factors $\boldsymbol{f}$ and 
$\boldsymbol{u}$ are unobserved, they must first be estimated. In the following subsection, we investigate the consistent estimation of these latent factors.

	\subsection{Factor estimation}
	Throughout this paper, we assume that the data $\{\boldsymbol{x}_{i}, Z_{i}, \boldsymbol{f}_{i}, \boldsymbol{u}_{i}, Y_{i}, \varepsilon_{i}\}_{i=1}^{n}$ are  independent and identically distributed (i.i.d.) copies of $\{\boldsymbol{x}, Z,  \boldsymbol{f}, \boldsymbol{u}, Y, \varepsilon\}$. {Let  $\boldsymbol{X}=(\boldsymbol{x}_{1}, \ldots, \boldsymbol{x}_{n})^{\top}$, $\boldsymbol{Z}=({Z}_{1}, \ldots, {Z}_{n})^{\top}$, $\boldsymbol{F}=(\boldsymbol{f}_{1}, \ldots, \boldsymbol{f}_{n})^{\top}$, $\boldsymbol{U}=(\boldsymbol{u}_{1}, \ldots, \boldsymbol{u}_{n})^{\top}$, $\boldsymbol{Y}=(Y_{1}, \ldots, Y_{n})^{\top}$ and $\boldsymbol{\varepsilon}=({\varepsilon}_{1},\ldots,{\varepsilon}_{n})^{\top}$.

  In the context of the FAPLM, the latent factor vectors $\boldsymbol{f}$ and $\boldsymbol{u}$ must be estimated, given that only the predictor vector 
		${\boldsymbol{x}}$ and the response $Y$ are directly observed.
        The constrained least squares estimator of $(\boldsymbol{F},\boldsymbol{B})$ based on ${\bX}$ is given by
		\begin{align}
			(\hat{\boldsymbol{F}},\hat{\boldsymbol{B}})=&\mathop{\arg\min}\limits_{\boldsymbol{F}\in \mathbb{R}^{n\times K},\boldsymbol{B}\in \mathbb{R}^{p\times K}}\|{\bX}-\boldsymbol{F}\boldsymbol{B}^{\top}\|_{\mathbb{F}}^{2} \\ \notag
			\text{subject} \ \  \text{to}& \  n^{-1}{\boldsymbol{F}}^{\top}\boldsymbol{F}=\boldsymbol{I}_{K} \ \ \text{and} \ \ {\boldsymbol{B}}^{\top}\boldsymbol{B} \ \ \text{is} \ \  \text{diagonal}. \notag
		\end{align}	
		Elementary manipulation yields that the columns of $\hat{\boldsymbol{F}}/\sqrt{n}$ are the eigenvectors corresponding to the largest $K$ eigenvalues of the matrix ${\bX}{{\bX}}^{\top}$ and 
				$\hat{\boldsymbol{B}}={\bX}^{\top}{\hat{\boldsymbol{F}}}({\hat{\boldsymbol{F}}}^{\top}\hat{\boldsymbol{F}})^{-1}=n^{-1}{\bX}^{\top}{\hat{\boldsymbol{F}}}$.
				Then the least squares estimator of $\boldsymbol{U}$ is given by $\hat{\boldsymbol{U}}={\bX}-\hat{\boldsymbol{F}}{\hat{\boldsymbol{B}}}^{\top}=(\boldsymbol{I}_{n}-n^{-1}\hat{\boldsymbol{F}}{\hat{\boldsymbol{F}}}^{\top}){\bX}$ \citep{fan2013large}. 
				Since $K$ corresponds to the number of spiked eigenvalues of ${\bX}{\bX}^{\top}$, it is generally small. Consequently, we treat $K$ as a fixed constant, as recommended by \cite{fanJ2023}.

				
				In practice, the number of latent factors, denoted  $K$, is often unknown, and determining $K$ in a data-driven way is a crucial challenge. Various methods have been proposed in the literature to estimate $K$ \citep{BaiandNg2002, Lam2012, Ahn2013, Fan2022}.
				In this paper, we adopt the ratio method for numerical studies.  Let $\lambda_{k}({\bX}{\bX}^{\top})$ be the $k$-th largest eigenvalue of ${\bX}{\bX}^{\top}$ and let $K_{\max}$ be a prescribed upper bound.
				The number of factors can
				be consistently estimated \citep{Luo2009, Lam2012, Ahn2013} by
				\begin{align}
					\hat{K}=\arg\max\limits_{k\leq K_{\max}}\frac{\lambda_{k}({\bX}{\bX}^{\top})}{\lambda_{k+1}({\bX}{\bX}^{\top})}, \notag
				\end{align}
				where $1 \leq K_{\max} \leq n$ is a prescribed upper bound for $K$. In  the subsequent theoretical analysis,  $K$ is treated as known. All theoretical results remain
valid, with the condition that $\hat{K}$ is a consistent estimator of $K$.

				\subsection{Regularization estimation} \label{Regularization Estimation}
           In this subsection,      we  consider the parameter estimation problem. 
				The regularized estimators of
				the unknown parameter vectors $\boldsymbol{\beta}_0$, $\boldsymbol{\varphi}_0$ and $\boldsymbol{\xi}_0$ 
				in \eqref{FAPLM_basis_approximation} are defined as follows,
				\begin{align} \label{empiricalestimator}
					(\hat{\boldsymbol{\beta}}, \hat{\boldsymbol{\varphi}}, \hat{\boldsymbol{\xi}})=\mathop{\arg\min}_{\boldsymbol{\beta}\in \mathbb{R}^{p}, \boldsymbol{\varphi}\in \mathbb{R}^{K}, \boldsymbol{\xi}\in \mathbb{R}^{M_{n}}} \left\{\frac{1}{2n}\left\|\boldsymbol{Y}-\hat{\boldsymbol{U}}\boldsymbol{\beta}-\hat{\boldsymbol{F}}\boldsymbol{\varphi}-\boldsymbol{\Pi}\boldsymbol{\xi}\right\|_{2}^{2}+\lambda\|\boldsymbol{\beta}\|_{1}\right\},
				\end{align}
				where $\boldsymbol{\Pi}=\{\boldsymbol{\Pi}(Z_{1}),\ldots,\boldsymbol{\Pi}(Z_{n})\}^{\top} \in  \mathbb{R}^{n\times M_{n}}$,  and 
				$\lambda >0 $ is a tuning parameter.
                Recall that $\hat{\boldsymbol{U}}=(\boldsymbol{I}_{n}-\hat{\boldsymbol{P}}){\bX}$,  where $\hat{\boldsymbol{P}}=n^{-1}\hat{\boldsymbol{F}}\hat{\boldsymbol{F}}^{\top}$ is the corresponding projection matrix. 
                Hence, $\hat{\boldsymbol{F}}^{\top}\hat{\boldsymbol{U}}={\boldsymbol{0}}$.
                It is straightforward to verify that the solution of \eqref{empiricalestimator} is
 				equivalent to
				\begin{align}(\hat{\boldsymbol{\beta}}, \hat{\boldsymbol{\xi}})&=\mathop{\arg\min}_{\boldsymbol{\beta}\in \mathbb{R}^{p}, \ \boldsymbol{\xi}\in \mathbb{R}^{M_{n}}}\left\{\frac{1}{2n}\left\|(\boldsymbol{I}_{n}-\hat{\boldsymbol{P}})(\boldsymbol{Y}-\boldsymbol{\Pi}\boldsymbol{\xi})-\hat{\boldsymbol{U}}\boldsymbol{\beta}\right\|_{2}^{2}+\lambda\|\boldsymbol{\beta}\|_{1}\right\}, \label{beta_phi_estimation} \\
                \hat{\boldsymbol{\varphi}}&=(\hat{\boldsymbol{F}}^{\top}\hat{\boldsymbol{F}})^{-1}\hat{\boldsymbol{F}}^{\top}(\boldsymbol{Y}-\hat{\boldsymbol{U}}\boldsymbol{\beta}-\boldsymbol{\Pi}\boldsymbol{\xi})=\frac{1}{n}\hat{\boldsymbol{F}}^{\top}(\boldsymbol{Y}-\boldsymbol{\Pi}\boldsymbol{\xi}).  \label{Least squares estimation of phi}
				\end{align}

				\section{Inference} \label{Statistical analysis}
				In this paper, we are interested in the problem of testing global parameters in model \eqref{FAPLM}, that is, the hypothesis testing problem is
				\begin{align}
					H_{0}: \boldsymbol{\beta}_0=\boldsymbol{0} \ \ \text{versus} \ \ H_{1}: \boldsymbol{\beta}_0\neq \boldsymbol{0}.
				\end{align}

Firstly, we correct the bias of the initial estimator $\hat{\boldsymbol{\beta}}$ in \eqref{beta_phi_estimation}. Building on the construction of the debiased estimator $\tilde{\boldsymbol{\beta}}$, we then develop a max-type test statistic and derive its asymptotic distribution. For statistical inference, we use the bootstrap method to obtain empirical critical values for hypothesis testing. 

				\subsection{Factor-adjusted projection debiased   test}

				In high-dimensional statistical inference, debiased Lasso estimators have been employed for  high-dimensional linear models, as developed by \cite{zhang2014confidence, van2014asymptotically, javanmard2014confidence}. 
				For the linear model $Y=\boldsymbol{X}\boldsymbol{\beta}_0+\boldsymbol{\varepsilon}$, given an initial Lasso estimator $\hat{\boldsymbol{\beta}}$ of $\boldsymbol{\beta}_0$, the debiased Lasso adds a correction term to $\hat{\beta}_{j}$  to remove the bias introduced by regularization. In particular, the correction term takes the form of 
				\begin{align} \label{non-partial debiased method}
					\frac{1}{n}\hat{\boldsymbol{\Gamma}}_{j}\boldsymbol{X}^{\top}(\boldsymbol{Y}-\boldsymbol{X}\hat{\boldsymbol{\beta}}),
				\end{align}
				where $n^{-1}\boldsymbol{X}^{\top}(\boldsymbol{Y}-\boldsymbol{X}\hat{\boldsymbol{\beta}})$ is the sample analogue of the population score function $\mathbb{E}\{\boldsymbol{x}_{i}^{\top}(Y_{i}-\boldsymbol{x}_{i}^{\top}\hat{\boldsymbol{\beta}})\}$, $\hat{\boldsymbol{\Gamma}}_{j}$ denotes the $j$-th row of $\hat{\boldsymbol{\Gamma}}$, which is an approximate inverse of $n^{-1}\boldsymbol{X}^{\top}\boldsymbol{X}$, whose population counterpart is $\mathbb{E}(\boldsymbol{x}_{i}\boldsymbol{x}_{i}^{\top})$. 
				

In our model \eqref{FAPLM}, the nuisance nonparametric component $g_{0}(\cdot)$ introduces additional bias. Consequently, the standard debiased Lasso cannot fully remove the influence of $g_{0}(\cdot)$, resulting in a limiting distribution that is not centered around zero. A commonly used solution is to project the parametric component onto the nonparametric space. 
\cite{zhu2017nonasymptotic} utilized the projection method to study the $\ell_1$-norm penalized estimator, demonstrating that the estimation error of the nonlinear component has only a second-order effect to the error of the linear component. This finding enables feasible inference based on the debiased estimator and  motivates us to apply similar projection strategy in analyzing the linear estimator
in FAPLM.

				
				Building on these methods, we propose
                the Factor-Adjusted Projection Debiased  Test (FAPDT). 
Specifically, we define the projection of ${U}_j, \ j=1,\dots,p$,  onto ${Z}$ as follows
	\begin{align}  
		{m}_{\mathcal{G}_{u},j}(Z)&=\mathop{\arg\min}\limits_{{m} \in \mathcal{G}_{u}} \mathbb{E}\left\{{U}_j-{m}(Z)\right\}^{2}.\label{projection u on z}
	\end{align}
	Here ${m}_{\mathcal{G}_{u},j}(Z)$   can be regarded as the projection of conditional expectation $\mathbb{E}({U}_j|Z)$  on the function space $\mathcal{G}_{u}$. In this paper, we approximate $\mathbb{E}({U}_j|Z)$  with B-spline
	method as well, that is, $\mathcal{G}_{u}=\left\{\boldsymbol{\Pi}( Z)^{\top}\boldsymbol{\xi}: \boldsymbol{\xi}\in \mathbb{R}^{M_{n}}\right\}$.
	Similar with the analysis of $g_{0}(\cdot)$, if $\mathbb{E}({U}_{j}|Z) \in \mathcal{H}_{r}$ for some $r \geq 1.5$,  it  can be well approximated by ${m}_{\mathcal{G}_{u},j}(Z)$  with sufficiently large $M_{n}$.  Denote
	$\boldsymbol{m}_{\mathcal{G}_{u}}(Z)=\left\{m_{\mathcal{G}_{u},1}(Z),\ldots,m_{\mathcal{G}_{u},p}(Z)\right\}^{\top} \in \mathbb{R}^p$,  ${{\boldsymbol{u}}}_i^{*}=\boldsymbol{u}_i-\mathbb{E}(\boldsymbol{u}_i|{Z}_i) \in \mathbb{R}^{p}, \ i=1,\ldots,n$ and  $\boldsymbol{\Sigma}_{u}^{*}=\Cov({{\boldsymbol{u}}}^{*})$. Let ${{\boldsymbol{U}}}^{*}=({{\boldsymbol{u}}}_1^{*},\ldots,{{\boldsymbol{u}}}_n^{*})^{\top} \in \mathbb{R}^{n \times p}$, $\boldsymbol{m}_{\mathcal{G}_{u}}(\boldsymbol{Z})=\{\boldsymbol{m}_{\mathcal{G}_{u}}(Z_1),\ldots,\boldsymbol{m}_{\mathcal{G}_{u}}(Z_n)\}^{\top}=\boldsymbol{\Pi}(\boldsymbol{\Pi}^{\top}\boldsymbol{\Pi})^{-1}\boldsymbol{\Pi}^{\top}\boldsymbol{U} \in \mathbb{R}^{n \times p}$, and  nonparametric surrogates  $\tilde{\boldsymbol{U}}=\hat{\boldsymbol{U}}-\hat{\boldsymbol{m}}_{\mathcal{G}_{u}}(\boldsymbol{Z})$, with
				$\hat{\boldsymbol{m}}_{\mathcal{G}_{u}}(\boldsymbol{Z})=\boldsymbol{\Pi}(\boldsymbol{\Pi}^{\top}\boldsymbol{\Pi})^{-1}\boldsymbol{\Pi}^{\top}\hat{\boldsymbol{U}}  \in \mathbb{R}^{n \times p}$. 

				Let ${\hat{\boldsymbol{\Theta}}} \in \mathbb{R}^{p\times p}$ be an approximation for the inverse of the Gram matrix $\tilde{\boldsymbol{\Sigma}}=n^{-1}\tilde{\boldsymbol{U}}^{\top}\tilde{\boldsymbol{U}}$ and $\boldsymbol{\Theta}=(\boldsymbol{\Sigma}_{u}^{*})^{-1}$.
				There are many ways to estimate such a precision matrix. 
				For instance, node-wise regression proposed by \cite{zhang2014confidence} and \cite{van2014asymptotically}, as well as the CLIME-type estimator provided in \cite{cai2011constrained}, \cite{javanmard2014confidence} and \cite{avella2018robust}. We construct ${\hat{\boldsymbol{\Theta}}}$ by employing a similar method   proposed by \cite{cai2011constrained}.  Concretely,  ${\hat{\boldsymbol{\Theta}}}$ is the solution to the following
				constrained optimization problem.
				\begin{align}\label{Theta hat estimation}
					\hat{\boldsymbol{\Theta}}&=\mathop{\arg\min}\limits_{\boldsymbol{\Theta} \in \mathbb{R}^{p \times p}}\|\boldsymbol{\Theta}\|_{1} \notag \\
					&\text{s.t.} \ \|\boldsymbol{\Theta}\tilde{\boldsymbol{\Sigma}}-\boldsymbol{I}_{p}\|_{\max}\leq \delta_{n},
				\end{align}
				where $\delta_{n}$ is a predetermined tuning parameter.
				In general, $\hat{\boldsymbol{\Theta}}$ is not symmetric since there is no symmetry constraint on \eqref{Theta hat estimation}. Symmetry can be enforced through additional operations. Denote $\hat{\boldsymbol{\Theta}}=(\hat{\Gamma}_{ij})_{1 \leq i,j \leq p}$. Write $\hat{\boldsymbol{\Theta}}^{\text{sym}}=(\hat{\Gamma}_{ij}^{\text{sym}})_{1 \leq i,j \leq p}$, where $\hat{\Gamma}_{ij}^{\text{sym}}$ is defined as 
				\begin{align}
					\hat{\Gamma}_{ij}^{\text{sym}}=\hat{\Gamma}_{ij}\mathbb{I}(|\hat{\Gamma}_{ij}| \leq |\hat{\Gamma}_{ji}|)+\hat{\Gamma}_{ji}\mathbb{I}(|\hat{\Gamma}_{ij}|> |\hat{\Gamma}_{ji}|). \notag 
				\end{align}
				Apparently, $\hat{\boldsymbol{\Theta}}^{\text{sym}}$ is a symmetric matrix. Without loss of generality, we assume that $\hat{\boldsymbol{\Theta}}$ is symmetric in the rest of the paper. 
				Given $\hat{\boldsymbol{\Theta}}$ and $\tilde{\boldsymbol{U}}$, the  Factor-adjusted  projection debiased estimator of $\boldsymbol{\beta}_0$ is then defined as
				\begin{align} \label{debiased estimator}
					\tilde{\boldsymbol{\beta}}=\hat{\boldsymbol{\beta}}+\frac{1}{n}\hat{\boldsymbol{\Theta}}\tilde{\boldsymbol{U}}^{\top}\left(\boldsymbol{Y}-\hat{\boldsymbol{U}}\hat{\boldsymbol{\beta}}-\hat{\boldsymbol{F}}\tilde{\boldsymbol{\varphi}}-\boldsymbol{\Pi}\hat{\boldsymbol{\xi}}\right),
				\end{align} 
				with  $\hat{\boldsymbol{\beta}}$ and $\hat{\boldsymbol{\xi}}$ defined in  \eqref{beta_phi_estimation}, and $	\tilde{\boldsymbol{\varphi}}=n^{-1}\hat{\boldsymbol{F}}^{\top}(\boldsymbol{Y}-\boldsymbol{\Pi}\hat{\boldsymbol{\xi}})$.
                
                In this paper, we consider the  max-type statistic $\sqrt{n}\|\tilde{\boldsymbol{\beta}}-\boldsymbol{\beta}_0\|_{\infty}$.
To determine its asymptotic distribution, we  define  $\boldsymbol{V}=(V_{1},\ldots,V_{p})^{\top}$ as a mean-zero Gaussian random vector with the same covariance matrix as $n^{-{1}/{2}}  \boldsymbol{\Theta}{\boldsymbol{U}}^{*\top}\boldsymbol{\varepsilon}$, that is, 
				\begin{align}
					\Cov(\boldsymbol{V})=\Cov\left(\frac{1}{\sqrt{n}} \boldsymbol{\Theta}{\boldsymbol{U}}^{*\top}\boldsymbol{\varepsilon}\right)={\sigma}_{\varepsilon}^{2}\boldsymbol{\Theta}.
				\end{align}
				Here ${\sigma}_{\varepsilon}^{2}$ is the variance of $\varepsilon$.
	\subsection{Gaussian multiplier bootstrap}
				Since ${\sigma}_{\varepsilon}$ and $\boldsymbol{\Theta}$ are unknown, obtaining the critical value $c_{1-\alpha}$ directly poses a challenge. \cite{fanJ2023} suggests using the Gaussian multiplier bootstrap to derive it. We provide an overview of the procedural steps.	
				\begin{itemize}
					\item [1.]     
					Generate i.i.d random variables $\upsilon_{1},\ldots, \upsilon_{n} \sim \mathrm{N}(0,1)$  independent of the observed dataset $\mathcal{D}=\{\boldsymbol{x}_{1},\ldots,\boldsymbol{x}_{n}, Y_{1},\ldots,Y_{n}\}$, and compute
					\begin{align} \label{bootstrap statistic}
						\hat{L}=\frac{1}{\sqrt{n}}\|\hat{\boldsymbol{\Theta}}\tilde{\boldsymbol{U}}^{\top}\boldsymbol{\upsilon}\|_{\infty},
					\end{align}
					where $\tilde{\boldsymbol{U}}=\hat{\boldsymbol{U}}-\hat{\boldsymbol{m}}_{\mathcal{G}_{u}}(\boldsymbol{Z})$, with $\hat{\boldsymbol{m}}_{\mathcal{G}_{u}}(\boldsymbol{Z})=\boldsymbol{\Pi}(\boldsymbol{\Pi}^{\top}\boldsymbol{\Pi})^{-1}\boldsymbol{\Pi}^{\top}\hat{\boldsymbol{U}}$,  and  $\boldsymbol{\upsilon}=\left(\upsilon_{1}, \ldots, \upsilon_{n}\right)^{\top}.$
					\item [2.]   Independently repeat the initial step $B$ times to obtain $\hat{L}_{1},\ldots,\hat{L}_{B}$. 
					Estimate the critical value $c_{1-\alpha}$
					via the $(1-\alpha)$-quantile of the empirical distribution of the bootstrap statistics:
					\begin{align} \label{critical value estimate}
						\hat{c}_{1-\alpha}=\inf \left\{t \geq0: H_{B}(t)=\frac{1}{B}\sum\limits_{b=1}^{B}\mathbb{I}(\hat{L}_{b} \leq t) \geq 1-\alpha\right\}.  
					\end{align} 
                    \end{itemize}
 Our test statistic for $H_{0}$ is given by
					\begin{align} \label{test statistic construction}
						T_{n}=\frac{\sqrt{n}\|\tilde{\boldsymbol{\beta}}\|_{\infty}}{\hat{\sigma}_{\varepsilon}},
					\end{align}
					where $\tilde{\boldsymbol{\beta}}$ is defined in \eqref{debiased estimator}, and $\hat{\sigma}_{\varepsilon}^{2}$ is defined as 
                    \begin{align} \label{sigma-epsilon hat}\hat{\sigma}_{\varepsilon}^{2}=\frac{1}{n}\sum_{i=1}^{n}(Y_{i}-\hat{\boldsymbol{u}}_{i}^{\top}\hat{\boldsymbol{\beta}}-\hat{\boldsymbol{f}}_{i}^{\top}\tilde{\boldsymbol{\varphi}}-\boldsymbol{\pi}_{i}^{\top}\hat{\boldsymbol{\xi}})^{2},
					\end{align}
with $\hat{\boldsymbol{\beta}}$ and $\hat{\boldsymbol{\xi}}$ being defined in \eqref{beta_phi_estimation}, and $	\tilde{\boldsymbol{\varphi}}=n^{-1}\hat{\boldsymbol{F}}^{\top}(\boldsymbol{Y}-\boldsymbol{\Pi}\hat{\boldsymbol{\xi}})$. We reject the null hypothesis $H_{0}$ if and only if  $
						T_{n} \geq \hat{c}_{1-\alpha}$.

\section{Theoretical results }\label{Theoretical Results}

In this section, we present the theoretical properties of the proposed estimation and inference procedures. To set the stage for the main results, we begin by outlining the necessary assumptions.
Firstly, we introduce the identifiability conditions similar to those in \cite{fan2013large}.
\begin{assumption} \label{basic asssumption in factor model}
			Assume that $\Cov(\boldsymbol{f}_{t})=\boldsymbol{I}_{K}$, $t=1,\ldots, n$, and $\boldsymbol{B}^{\top}\boldsymbol{B}$ is diagonal.
		\end{assumption}
More specifically, let $\hat{\boldsymbol{\Sigma}}$, $\hat{\boldsymbol{\Lambda}}=\text{diag}(\hat{\lambda}_{1},\ldots, \hat{\lambda}_{K})$ and $\hat{\boldsymbol{\Omega}}=\left(\hat{\boldsymbol{\eta}}_{1},\ldots, \hat{\boldsymbol{\eta}}_{K}\right)$ be initial pilot estimators of the covariance matrix $\boldsymbol{\Sigma}$ of $\boldsymbol{x}$, the matrix consisting of its leading  $K$ eigenvalues $\boldsymbol{\Lambda}=\text{diag}(\lambda_{1},\ldots, \lambda_{K})$, and the matrix consisting of  their corresponding leading $K$ orthonormalized eigenvectors $\boldsymbol{\Omega}=(\boldsymbol{\eta}_{1},\ldots, \boldsymbol{\eta}_{K})$, respectively.
				We
				adopt the regularity assumptions in \cite{fanJ2023, fan2013large,Bai2003} and other literature on high-dimensional factor analysis.
\begin{assumption} \label{factorassumption} 
					We make the following assumptions. 		
					\begin{itemize}
						\item[(i)] There exists some positive constant $c_{0} < \infty$ such that $\|\boldsymbol{f}\|_{\psi_{2}}\leq c_{0}$ and $\|\boldsymbol{u}\|_{\psi_{2}}\leq c_{0}$. In addition, $\mathbb{E}(U_{ij})=\mathbb{E}(F_{ik})=0, \ i=1,\ldots,n, \ j=1,\ldots,p, \  k=1,\ldots,K$.
						\item[(ii)] There exists a constant $\tau >1$ such that $p/\tau\leq \lambda_{min}(\boldsymbol{B}^{\top}\boldsymbol{B})\leq \lambda_{\max}(\boldsymbol{B}^{\top}\boldsymbol{B})\leq p\tau$. 
						\item[(iii)]  Let $\boldsymbol{\Sigma}_{u}=\Cov(\boldsymbol{u})$, $\|\boldsymbol{\Sigma}_{u}\|_{2}$ is bounded. In addition,  there exists a constant $\Upsilon>0$ such that $\|\boldsymbol{B}\|_{\max}\leq \Upsilon$ and $$\mathbb{E}|\boldsymbol{u}^{\top}\boldsymbol{u}-\tr(\boldsymbol{\Sigma}_{u})|^{4}\leq \Upsilon p^{2}.$$ 
						\item[(iv)] There exists a positive constant $\kappa<1$ such that $\kappa \leq \lambda_{min}(\boldsymbol{\Sigma}_{u})$, $\|\boldsymbol{\Sigma}_{u}\|_{1}\leq 1/\kappa$ and $\min\limits_{1\leq k,l\leq p} \Var({U}_{ik}{U}_{il})\geq \kappa$. \\
					\end{itemize}
				\end{assumption}
				\begin{assumption} (Initial pilot estimators). \label{Loadings and initial pilot estimators}
					Assume that 
					$\hat{\boldsymbol{\Sigma}}$, $\hat{\boldsymbol{\Lambda}}$ and $\hat{\boldsymbol{\Omega}}$ satisfy $\|\hat{\boldsymbol{\Sigma}}-\boldsymbol{\Sigma}\|_{\max}=O_{\mathbb{P}}\{\sqrt{(\log p)/n}\}$, $\|(\hat{\boldsymbol{\Lambda}}-\boldsymbol{\Lambda})\hat{\boldsymbol{\Lambda}}^{-1}\|_{\max}=O_{\mathbb{P}}\{\sqrt{(\log p)/n}\}$, and $\|\hat{\boldsymbol{\Omega}}-\boldsymbol{\Omega}\|_{\max}=O_{\mathbb{P}}\{\sqrt{(\log p)/(np)}\}$.
				\end{assumption}
				\begin{remark}
					Assumption \ref{Loadings and initial pilot estimators} is adopted from \cite{bayle2022factor} and is applicable in numerous scenarios, including the sample covariance matrix under sub-Gaussian distributions \citep{fan2013large}. Furthermore, estimators such as the marginal and spatial Kendall’s tau \citep{fan2018large} and the elementwise adaptive Huber estimator \citep{fan2019robust} are consistent with this assumption.
				\end{remark}

				We summarize the theoretical results related to consistent factor estimation in Lemma B.1 in Supplementary Material,  which directly follows from Proposition 2.1 in \cite{fanJ2023} and Lemma 3.1 in \cite{bayle2022factor}.

                \subsection{Estimation error}
				
In the following, we investigate the error bounds of the estimators in \eqref{beta_phi_estimation} and analyze how factor estimation and nonparametric approximation  affect parameter estimation. 
				Recall that ${{\boldsymbol{U}}}^{*}=\boldsymbol{U}-\mathbb{E}(\boldsymbol{U}|\boldsymbol{Z})$,  $\tilde{\boldsymbol{U}}=\hat{\boldsymbol{U}}-\hat{\boldsymbol{m}}_{\mathcal{G}_{u}}(\boldsymbol{Z})$, and  $\boldsymbol{\Sigma}_{u}^{*}=\Cov({{\boldsymbol{u}}}^{*})$. To investigate the consistency property of $\hat{\boldsymbol{\beta}}$, we need to make the following technical conditions.  
				
				\begin{assumption} \label{consistency property assumption}
					We make the following assumptions.
					\begin{itemize}
						\item[(i)]   There exist  positive constants $c_{1}, c_{u} < \infty$, such that $\|{\varepsilon}\|_{\psi_{2}}\leq c_{1}$ and $\|{{\boldsymbol{u}}}^*\|_{\psi_{2}} \leq  c_{u}$.
						\item[(ii)] The covariance matrix $\boldsymbol{\Sigma}=\Cov(\boldsymbol{x})>0$, and   $\lambda_{\min}(\boldsymbol{\Sigma}_{u}^{*})$ is bounded away from zero.
						\item[(iii)] There exists  some constant  $r\geq1.5$, such that $g_{0} \in \mathcal{H}_{r}$, and each component of  $\mathbb{E}(\boldsymbol{u}|Z)$  belongs to $\mathcal{H}_{r}$. 
						\item[(iv)] The variable $Z$ is bounded as $\sup |Z|< \infty$. Without loss of generality, we assume that $Z \in [0,1]$.
					\end{itemize}
				\end{assumption}
Assumption \ref{consistency property assumption} is mild and commonly adopted in high-dimensional analysis. 
				Specifically, condition (i) imposes the sub-Gaussian assumption on $\varepsilon$, consistent with \cite{fanJ2023}. Condition (ii) requires the covariance matrix of $\boldsymbol{x}$ to be positive definite and places a lower bound on the minimum eigenvalue of $\boldsymbol{\Sigma}_{u}^{*}$, both of which are standard requirements in high-dimensional literature \citep{li2022integrative, ouyang2023high, fanJ2023}. Condition (iii) is a standard assumption on the smoothness of the function for the B-spline method \citep{xue2006additive, sherwood2016partially}. Finally, condition (iv) assumes that the variable $Z$ takes values within a compact interval $[0, 1]$, as also assumed in \cite{xie2009scad, lian2019projected}.

				We next  analyze the estimation consistency of 
				$\hat{\boldsymbol{\beta}}$. Denote  $\mathcal{S}^{*}=\{j \in [p]: {\beta}_{j} \neq 0\}$,    $s=\left|\mathcal{S}^{*}\right|$  is its cardinality, and $\mathcal{S}^{*C}=[p]/\mathcal{S}^{*}$.
				\begin{theorem} \label{consistency property}
					Suppose that the true parameter $\boldsymbol{\beta}_{0}$ satisfies that $\|\boldsymbol{\beta}_{0}\|_{\infty} \leq C$ and $s=o\{{\sqrt{n}}/{(\log p)}+\sqrt{{n}/(M_{n}\log p)}\}$. Additionally, we assume that				
					$(\log p)^{12}=O({n})$.
					By taking $\lambda=O\{\sqrt{(\log p)/n}\}$ and  $M_{n}\approx n^{1/(2r+1)}$, under Assumptions \ref{basic asssumption in factor model}-\ref{consistency property assumption}, the $\hat{\boldsymbol{\beta}}$ defined in \eqref{beta_phi_estimation} satisfies that
					\begin{align}
						\|\hat{\boldsymbol{\beta}}-\boldsymbol{\beta}_{0}\|_{2}=O_{\mathbb{P}} \left(\lambda \sqrt{s}\right), \ \ and \ \ \|\hat{\boldsymbol{\beta}}-\boldsymbol{\beta}_{0}\|_{1}=O_{\mathbb{P}} \left(\lambda {s}\right). \notag 
					\end{align}
				\end{theorem}
				
				\begin{remark}
					Theorem \ref{consistency property} establishes the $\ell_{1}$- and $\ell_{2}$-norm error bounds for our estimator  $\hat{\boldsymbol{\beta}}$.  The order   $\sqrt{s(\log p)/n}$ agrees with the minimax rate of the standard Lasso problem \citep{raskutti2011minimax}. 
                    From the proof of Theorem \ref{consistency property},  it is evident that the contribution of the nonlinear approximation error and the latent factor estimation error to the linear component's estimation error is
					\begin{align*}
						\sqrt{s}M_n^{-2r}(\log p)(\log n)+\frac{\sqrt{s}(\log p)^{3/2}(\log n)^2}{n}+M_n^{-2r}\sqrt{sM_n(\log p)},
					\end{align*}
					which is trivial compared to $\sqrt{s(\log p)/n}$ under the conditions outlined in Theorem \ref{consistency property}.  
				\end{remark}

Using the estimated parameter $\hat{\boldsymbol{\xi}}$ from \eqref{beta_phi_estimation}, an estimator of the nonparametric function $g_0(Z)$ can be constructed as $\hat{g}(Z) = \boldsymbol{\pi}(Z)^{\top}\hat{\boldsymbol{\xi}}$.
                The error bound for this estimator is given by the following Proposition.
\begin{proposition} (Error bound of nonlinear part). \label{g_hat error bound proposition}
    Suppose that the conditions in Theorem \ref{consistency property}  hold. The estimator $\hat{g}(Z)=\boldsymbol{\pi}(Z)^{\top}\hat{\boldsymbol{\xi}}$ satisfies that 
    \begin{align}
    {\frac{1}{n}\sum\limits_{i=1}^{n}\left\{\hat{g}(Z_i)-g_0(Z_i)\right\}^2 }=O_{\mathbb{P}}\left\{M_n^{-2r}+{\frac{({\log p})(\log n)^{2}}{{n}}}+s^2\frac{(\log p)^2}{{n}}\right\}. \notag 
    \end{align}
\end{proposition}
Proposition \ref{g_hat error bound proposition} provides the error bound for the nonlinear estimator. This result demonstrates that the estimation error of the latent factors does not affect the order of the nonlinear component's estimation error, provided that $(\log p)^3 \ll M_n$.
				%
		\subsection{Asymptotic properties of  factor-adjusted projection debiased estimator}		

To establish an accuracy bound for $\hat{\boldsymbol{\Theta}}$ as an approximation of the inverse of the pseudo Hessian matrix $\tilde{\boldsymbol{\Sigma}}$, we first introduce a necessary assumption regarding the sparsity of $\boldsymbol{\Theta}$.

	\begin{assumption}\label{assumption on the sparsity of the inverse of SIgmau}
					There exists some positive constant $M_{\Sigma}$, such that $\left\|\boldsymbol{\Theta}\right\|_{\infty} \leq M_{\Sigma}$. Moreover,  $\boldsymbol{\Theta}=\left({\boldsymbol{\Gamma}}_{1},\ldots,{\boldsymbol{\Gamma}}_{p}\right)^{\top}=({\Gamma}_{ij})_{1 \leq i,j \leq p}$ is  row-wise sparse, i.e., $\max_{i \in [p]}\sum_{j=1}^{p}|{\Gamma}_{ij}|^{q} \leq c_{n,p}$, for some $0 \leq q < 1$, where $c_{n,p}$ is positive and bounded away from zero and allowed to increase as $n$ and $p$ grow.
				\end{assumption}
				
				Assumption \ref{assumption on the sparsity of the inverse of SIgmau} 	requires that  $\boldsymbol{\Theta}$  is sparse with respect to both its $\ell_{\infty}$-norm and its matrix row space.  Similar assumptions regarding precision matrix estimation and more general inverse Hessian matrix estimation have been discussed by \cite{van2014asymptotically}, \cite{cai2016estimating}, and \cite{ning2017general}. The estimation error bounds for $\hat{\boldsymbol{\Theta}}$ and the upper bound of $\|\hat{\boldsymbol{\Theta}} - \boldsymbol{\Theta}\|_{\infty}$ are provided in Proposition B.1 and Lemma B.9 in the Supplementary Material.
                These results are essential for establishing the theoretical foundations.			
				In the following, we give the Gaussian approximation result for the test statistic.

				%
				
				\begin{theorem} \label{Gaussianapproximationtheoremresult}
					Suppose that the conditions in Theorem \ref{consistency property} and Assumption \ref{assumption on the sparsity of the inverse of SIgmau} are satisfied.
                    Additionally, we assume that  
					$s=o[{{\sqrt{n}}/{(\log p)^{2}}}+\sqrt{n/\{M_{n}(\log p)^{3}\}}]$ and $(\log p)^{(12-4q)/(3-3q)}=o(n)$, with $0 \leq q <1$ being the constant in Assumption \ref{assumption on the sparsity of the inverse of SIgmau}. Then, under the null hypothesis, we have  
					\begin{align}
						\sup\limits_{ x>0}\left| \Pr\left(\sqrt{n}\|\tilde{\boldsymbol{\beta}}-\boldsymbol{\beta}_0\|_{\infty}\leq x\right)-\Pr\left(\|\boldsymbol{V}\|_{\infty} \leq x\right)\right| \rightarrow 0.
					\end{align}
				\end{theorem}
				
				Theorem \ref{Gaussianapproximationtheoremresult} indicates that our test statistic can be approximated by the maximum of a high-dimensional mean-zero Gaussian vector under mild conditions.  Based on this,  we can reject the null hypothesis $H_{0}$ at the significant level $\alpha$ if and only if $\sqrt{n}\|\tilde{\boldsymbol{\beta}}-\boldsymbol{\beta}_0\|_{\infty} > c_{1-\alpha}$, where $c_{1-\alpha}$ is the $(1-\alpha)$-th quantile of the distribution of $\|\boldsymbol{V}\|_{\infty}$.

We next demonstrate that the estimator $\hat{\sigma}_{\varepsilon}^{2}$, defined in \eqref{sigma-epsilon hat}, is consistent to the true variance ${\sigma}_{\varepsilon}^{2}$.
				\begin{proposition} \label{sigma estimation consistency lemma}
					Suppose that the conditions in Theorem \ref{consistency property} hold
					and  $s=o(n^{1/4}/\sqrt{\log p})$. Then, under the null hypothesis,  we have 
					\begin{align}
						\hat{\sigma}_{\varepsilon}^{2}-\sigma_{\varepsilon}^{2}=o_{\mathbb{P}}\left(\frac{1}{\log p}\right). \notag
					\end{align}
				\end{proposition}
				\begin{remark}
					Proposition \ref{sigma estimation consistency lemma} establishes the consistency of the variance estimator $\hat{\sigma}_{\varepsilon}^{2}$.
					Actually, the estimator of ${\sigma}_{\varepsilon}^{2}$ in high-dimensional linear regression has been extensively studied in the literature. For instance, \cite{fan2012variance} proposed the refitted cross-validation method to develop a consistent estimation that quantifies the uncertainty of 
					$\hat{\sigma}_{\varepsilon}$ 	in ultra-high-dimensional settings.
					Moreover, \cite{sun2012scaled} introduced the scaled-Lasso method, while \cite{yu2019estimating} developed the organic Lasso  method for estimating ${\sigma}_{\varepsilon}$.

				\end{remark}

				
				We  give the theoretical results related to the validity of the  bootstrap procedure in the following Theorem.
				\begin{theorem} \label{the validity of the  bootstrap procedure}
					Suppose that  the conditions in  Theorem \ref{Gaussianapproximationtheoremresult} and Proposition \ref{sigma estimation consistency lemma} hold. 
					Then, under the null hypothesis, we have 
					\begin{align}
						\sup\limits_{x>0}\left|\Pr\left(\frac{{\sqrt{n}\|\tilde{\boldsymbol{\beta}}-{\boldsymbol{\beta}}_0\|_{\infty}}}{{\hat{\sigma}_{\varepsilon}}}
						\leq x\right)-\Pr^{*}(\hat{L} \leq x)\right|\rightarrow 0, \notag 
					\end{align}
					where $\Pr^{*}(\cdot)=\Pr(\cdot|\mathcal{D})$ denotes the conditional probability.
				\end{theorem}
				
				Theorem \ref{the validity of the bootstrap procedure} establishes the validity of the proposed bootstrap procedure, based on the Gaussian approximation theory outlined in Theorem \ref{Gaussianapproximationtheoremresult}. Building on this foundation, we can define our decision rule as follows
				\begin{align*} 
					\psi_{\infty, \alpha}=\mathbb{I}\left(T_{n} >  \hat{c}_{1-\alpha}\right). 
				\end{align*} 
				The null hypothesis $H_{0}$ is rejected if and only if $\psi_{\infty, \alpha}=1$.

            Finally,    we consider the power  performance.
To demonstrate the validity of the test, we consider the following   local alternative hypotheses family for $\boldsymbol{\beta}$.
\begin{align} \label{local alternative hypotheses family}
\mathcal{B}(C)=\left\{\boldsymbol{\beta} \in \mathbb{R}^{p}: \max \limits_{j \in [p]}|\beta_j|\geq \sqrt{C\frac{\log p}{n}}\right\},
\end{align}
where  $C$ is a positive constant.
Next, we present the power property of the test statistic $T_n$ under the alternative hypothesis $H_1$.
\begin{theorem} \label{power theorem}
Suppose that the conditions in Theorem \ref{the validity of the  bootstrap procedure} are hold. For the test statistic $T_n$ defined in \eqref{test statistic construction}, we have 
\begin{align}
\lim\limits_{(n,p)\rightarrow \infty} \inf\limits_{\boldsymbol{\beta}  \in  \mathcal{B}(C) }\Pr\left( T_n \geq \hat{c}_{1-\alpha}\right)=1,
\end{align}
where $\mathcal{B}(C)$ is the  local
alternative hypotheses family defined in \eqref{local alternative hypotheses family}.
\end{theorem}

Theorem \ref{power theorem} demonstrates that our testing procedure maintains high power even when only a single component of the parameter $\boldsymbol{\beta}$ exceeds the order of $\sqrt{C(\log p)/n}$. This result highlights the effectiveness of our method under sparse alternative hypotheses, making it  powerful in such scenarios. Moreover, this separation rate aligns with the minimax optimal rate for detecting local alternative hypotheses, as established in prior works by \cite{verzelen2012minimax}, \cite{tony2014two}, \cite{zhang2017simultaneous} and \cite{ma2021global}.

				\section{Simulations}\label{Numerical studies section}
				
				In this section, we conduct simulation studies to assess the finite sample performance of the proposed method.

				\subsection{Estimation performance}
				\label{Estimation performance}
                Throughout the simulation study, we generate data based on the following  model:
				\begin{align}
					\boldsymbol{Y}&=\boldsymbol{U}\boldsymbol{\beta}_0+\boldsymbol{F}\boldsymbol{\phi}_0+g_{0}(\boldsymbol{Z})+\boldsymbol{\varepsilon}, \notag \\
					\text{with} \ \ \boldsymbol{X}&=\boldsymbol{F}\boldsymbol{B}^{\top}+\boldsymbol{U}. \notag
				\end{align}
                Here  $\boldsymbol{\phi}_0=\boldsymbol{B}^{\top}\boldsymbol{\beta}_0$.
                  To illustrate the accuracy of our estimation, we set  the number of latent factors to $K=2$, and the dimension of $\boldsymbol{x}$ to either $p=200$ or $500$. The first $s=5$  entries of $\boldsymbol{\beta}_0$ are set  to 2,  while the remaining $p-s$ entries are set to 0. 
				Throughout this subsection, we generate each entry of  $\boldsymbol{F}$ and $\boldsymbol{U}$ from the standard Gaussian distribution $\mathrm{N}(0, 1)$, and each entry of $\boldsymbol{B}$ 
				is randomly drawn from the uniform distribution $\mathrm{Unif}(-1, 1)$.
                The error term $\varepsilon$ is generated from   the Gaussian distribution $\mathrm{N}(0, 0.25)$.  In addition, we consider the following  two forms of 
				$g_0(Z_i), i=1,\ldots,n$.
				\begin{align}
					&\textbf{Model 1}: 
					g_0(Z_i)= Z_i; \label{Model 1 simulation} \\
					&\textbf{Model 2}: 
					g_0(Z_i)=\sin ( 2\pi Z_i). \label{Model 3 simulation}
				\end{align}
                The variable $Z$ is generated from the uniform distribution $\mathrm{Unif}(0, 1)$. 
				Note that Model \eqref{Model 1 simulation} corresponds to a linear model, while Model  \eqref{Model 3 simulation} corresponds to a nonlinear model.
			Under each setting, we generate $n=200$ i.i.d. observations, and replicate 500 times. 
                Following the recommendation of \cite{lian2019projected}, we use cubic splines (of order 4) to approximate the nonparametric function $g_0(\cdot)$ and set the number of internal nodes to $n^{1/9}$, which is the theoretical
				optimal order.

To emphasize the importance of incorporating latent factors and partially linear structures in modeling the relationships between response variables and covariates, we compare our method (denoted as ``FAPLM'') with two existing approaches.
The first is the method of \cite{lian2019projected} (denoted as ``PLM''), which uses B-spline method  to approximate  the nonparametric function and then applies the standard Lasso method with partial penalties for parameter estimation, but without accounting for latent factor effects. 
The second method is that of \cite{fanJ2023}, referred to as “FALM”, which estimates parameters under the assumption of linear relationships between  predictors and the response $Y$. This approach integrates 
$\boldsymbol{x}$ and ${Z}$, while simultaneously considering their factor decomposition.
We measure the prediction accuracy
by using the following indicators.
\begin{itemize}
    \item $\ell_1$-norm error: $\|\hat{\boldsymbol{\beta}}-{\boldsymbol{\beta}}\|_1$.
     \item Root Mean Squared Error (RMSE): $n^{-1/2}[\sum_{i=1}^n\{\hat{g}(Z_i)-g_0(Z_i)\}^2]^{1/2}$.
\end{itemize}
Table \ref{accuary table1} presents the estimation results for FAPLM, PLM, and FALM under the linear model corresponding to the link function \eqref{Model 1 simulation}. The results show that, for both $p=200$ and $p=500$, FALM outperforms the other methods in terms of both 
$\ell_1$-norm error and RMSE, indicating that FALM provides the most accurate estimation. This is because the model in these settings corresponds to the true model of \cite{fanJ2023}. In contrast, our method still relies on B-splines to estimate the nonparametric part, which may introduce approximation errors. However, compared to PLM, our method demonstrates superior performance in both $\ell_1$-norm error and RMSE, highlighting the advantage of incorporating factor analysis into partially linear models for more precise estimation.
Table \ref{accuary table2} presents the estimation results for FAPLM, PLM, and FALM under the nonlinear model corresponding to the link function \eqref{Model 3 simulation}, with results differing significantly from those in the linear case. For both  $p=200$ and $p=500$, our method achieves the best performance in terms of $\ell_1$-norm error and RMSE, while FALM performs the worst. This is because, for the nonlinear form of the nonparametric function, the method of \cite{fanJ2023} incorrectly assumes linearity and applies factor analysis and parameter estimation based on this assumption, leading to substantial errors. In contrast, our method uses B-splines to approximate the nonparametric functions, resulting in a smaller approximation error and more accurate estimation. Furthermore, compared to PLM, our method benefits from the factor model's ability to more effectively capture feature dependencies, resulting in improved estimation performance.

				\begin{table}[H]
		\small
		
		\renewcommand\arraystretch{1.1}
		\centering \tabcolsep 12pt \LTcapwidth 6in
		\caption{Empirical average of all indicators in linear model \eqref{Model 1 simulation}  based on 500 replications. The ``FAPLM'', ``PLM'' and   ``FALM'' represent the estimation results under FAPLM in this paper and PLM  without incorporating the factor effect in \cite{lian2019projected}, and  linear regression model accounting for latent factor effects in \cite{fanJ2023},  respectively.
		}
		\label{accuary table1}
		\begin{threeparttable}
			\begin{tabular}{ccccc}
				\toprule
				Method &$p$& FAPLM    & PLM& FALM    \\ \midrule
				$\ell_1$-norm error&\multirow{2}{*}{$200$}&\multirow{1}{*}{0.675}       & 0.724     &  0.566         \\ 
				RMSE& &\multirow{1}{*}{0.085}           &  0.096    & 0.073       \\
                \hline
				
				$\ell_1$-norm error&\multirow{2}{*}{$500$}&\multirow{1}{*}{0.796}       & 0.860     &  0.630   \\   
				RMSE& &\multirow{1}{*}{0.087}           &  0.099     & 0.076    \\ 
				\bottomrule
			\end{tabular}
		\end{threeparttable}
	\end{table}

		\begin{table}[H]
		\small
		
		\renewcommand\arraystretch{1.1}
		\centering \tabcolsep 12pt \LTcapwidth 6in
		\caption{Empirical average of all indicators in nonlinear model \eqref{Model 3 simulation}  based on 500 replications. The ``FAPLM'', ``PLM'' and   ``FALM'' represent the estimation results under FAPLM in this paper and PLM  without incorporating the factor effect  in \cite{lian2019projected}, and linear regression model accounting for latent factor effects in \cite{fanJ2023},  respectively.
		}
		\label{accuary table2}
		\begin{threeparttable}
			\begin{tabular}{ccccc}
				\toprule
				Method &$p$& FAPLM    & PLM& FALM    \\ \midrule
				$\ell_1$-norm error&\multirow{2}{*}{$200$}&\multirow{1}{*}{0.636}       & 0.903     &  0.912         \\ 
				RMSE& &\multirow{1}{*}{0.096}           &  0.117    & 0.660       \\
                \hline
				
				$\ell_1$-norm error&\multirow{2}{*}{$500$}&\multirow{1}{*}{0.772}       & 0.899     &  1.015   \\   
				RMSE& &\multirow{1}{*}{0.096}           &  0.113     & 0.659    \\ 
				\bottomrule
			\end{tabular}
		\end{threeparttable}
	\end{table}

				\subsection{Factor-adjusted projection debiased Lasso inference}
				In this subsection,  we set $K=2$, $p=200$ or $500$.
				Here  the entries of $\boldsymbol{B}$ are  generated from the uniform  distribution $\mathrm{Unif}(-1,1)$,  every row of $\boldsymbol{F}$ follows from $\mathrm{N}(\boldsymbol{0},\boldsymbol{I}_{K})$, and   every row of  $\boldsymbol{U}$ follows from $ \mathrm{N}(\boldsymbol{0},\boldsymbol{I}_{p})$.				
				We set ${\bX}=\boldsymbol{F}\boldsymbol{B}^{\top}+\boldsymbol{U}$.
                The distributions of the variable 
$Z$ and the error term $\varepsilon$, and the forms of $g_0(\cdot)$ are the same as those in Section \ref{Estimation performance}.
				We specify the parameter vector $\boldsymbol{\beta}_0=\omega\ast\left(\mathbf{1}_{5},\mathbf{0}_{p-5}\right)^{\top}$, $\omega \geq 0$.
				When $\omega=0$, it indicates that the null hypothesis holds, and the simulation results correspond to the empirical size. Otherwise, they correspond to the empirical power.

                To illustrate the importance of accounting for latent factors and leveraging PLM to capture the partially linear structure between the response variable and covariates, we compare our method with the testing approaches of \cite{zhu2019high} (denoted as ``PDT") and \cite{fanJ2023} (denoted as ``FabT"). While PDT employs projection and debiased techniques within PLMs to test parameters, it overlooks latent factors. In contrast, FabT incorporates latent factors but assumes a fully linear structure between the response and covariates.
The results for dimensions $p=200$ and $p=500$  are presented in Figures \ref{p_200_test} and \ref{p_500_test}, respectively. 
The label ``FAPDT" corresponds to the method proposed in this paper.
When $p=200$, for both Models \eqref{Model 1 simulation} and \eqref{Model 3 simulation},  when $\omega=0$, all three tests exhibit size performance close to the nominal significance level of 0.05.  As $\omega$ increases, the power curves of FAPDT and FabT for the linear model show similar convergence rates. However, for the nonlinear model, the power curve of FAPDT increases at a faster rate than that of FabT. This highlights the advantage of using a PLM to more effectively capture the partially linear structure between the response and covariates. In addition, for both Models \eqref{Model 1 simulation} and \eqref{Model 3 simulation}, the power convergence rates of FAPDT and FabT consistently outperform those of PDT. This underscores the critical role of considering latent factors to enhance model performance. As the parameter dimension increases to $p=500$, the patterns remain consistent, further reinforcing the effectiveness of FAPDT and the importance of accounting for latent factors.

				\begin{figure}[]%
					\centering
					\includegraphics[width=0.9\textwidth]{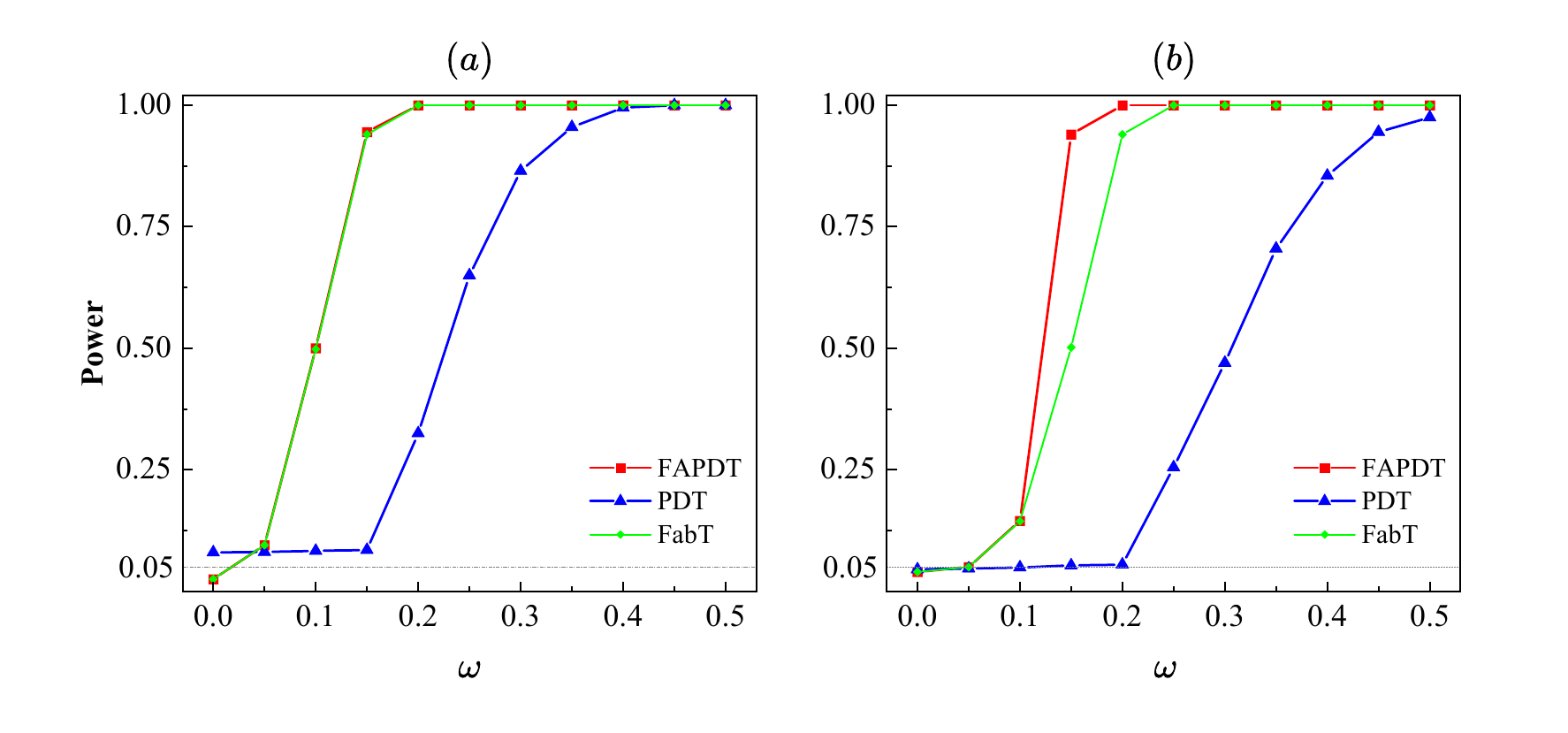}
					\caption{Power curves of linear model \eqref{Model 1 simulation} (the left panel) and nonlinear model \eqref{Model 3 simulation} (the right panel) with $p=200$ and $\boldsymbol{\beta}_0=\omega\ast\left(\mathbf{1}_{5},\mathbf{0}_{p-5}\right)^{\top}$.
						The  ``FAPDT'',  ``PDT''   and ``FabT''  signify the results derived from the method in this paper   and  the approach in \cite{zhu2019high} and \cite{fanJ2023}, respectively. } \label{p_200_test}
				\end{figure}

				\begin{figure}[]%
					\centering
					\includegraphics[width=0.9\textwidth]{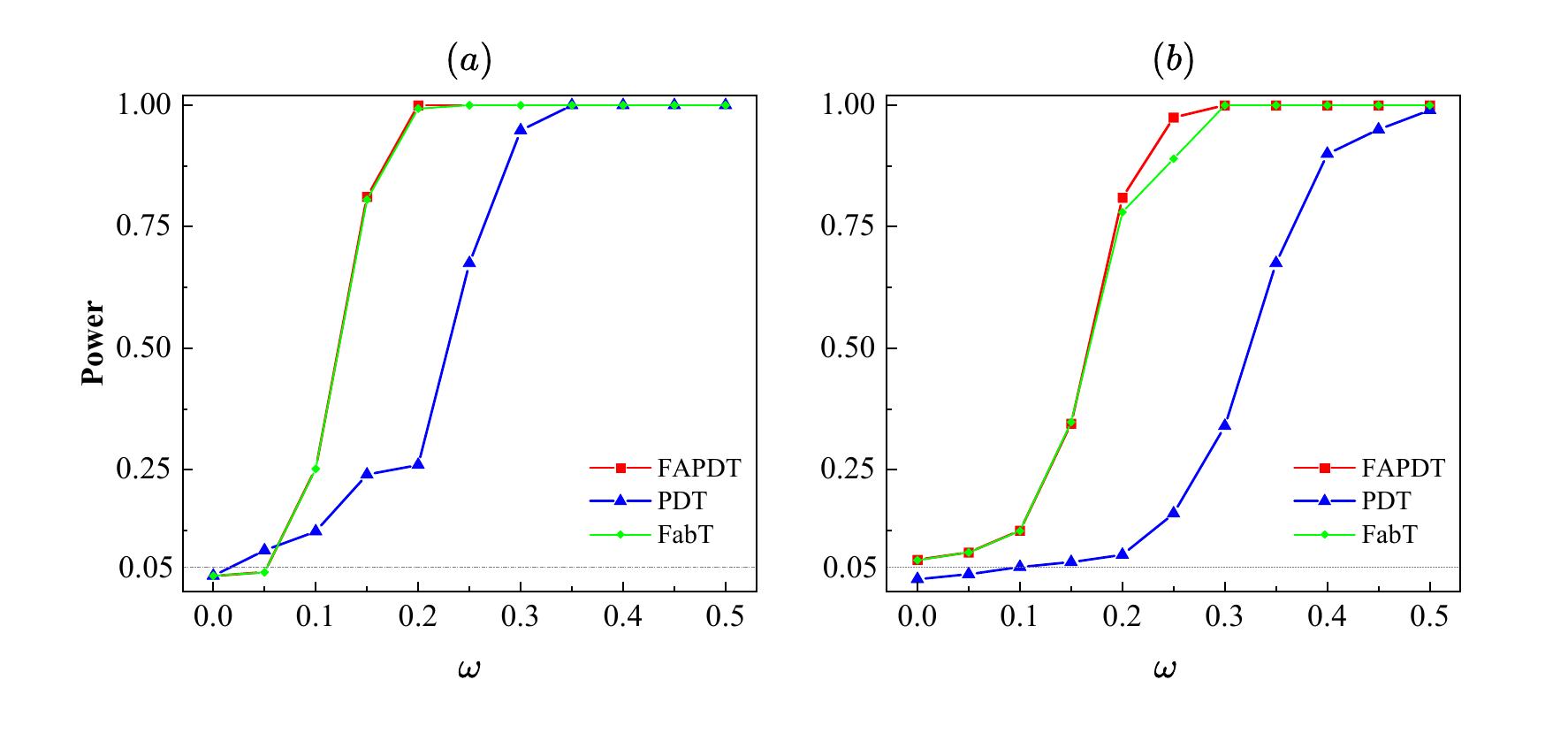}
					\caption{Power curves of linear model \eqref{Model 1 simulation} (the left panel) and nonlinear model \eqref{Model 3 simulation} (the right panel) with $p=500$ and $\boldsymbol{\beta}_0=\omega\ast\left(\mathbf{1}_{5},\mathbf{0}_{p-5}\right)^{\top}$.
						The  ``FAPDT'',  ``PDT''   and ``FabT''  signify the results derived from the method in this paper   and  the approach in \cite{zhu2019high} and \cite{fanJ2023}, respectively.} \label{p_500_test}
				\end{figure}

				\section{Real data analysis}\label{Real data analysis section}
In this section, we apply the proposed methods to analyze the birth weight data \citep{votavova2011transcriptome}, which includes blood samples and gene expression profiles collected from 72 pregnant women. Along with infant birth weights, maternal characteristics such as age, BMI, gestational age, parity, and blood cotinine levels are recorded. 
Low birth weight is strongly associated with increased risks of morbidity and mortality in newborns \citep{mcintire1999birth} and is further linked to a higher likelihood of developing chronic conditions later in life, such as obesity and type II diabetes \citep{ishida2012maternal}. Fetal growth, which is largely determined by the nutrients delivered through the placenta, is influenced by both environmental factors and gene expression \citep{ishida2012maternal}. 
Building on these findings, it is crucial to investigate the relationship between maternal gene expression and low birth weight. To effectively capture the partially linear structure between the response variable and covariates, while accounting for the significant correlations among gene expression profiles discussed in Section \ref{Introduction}, we employ the FAPLM for robust estimation and inference.

After excluding incomplete cases, a total of 65 individuals are retained, each with expression data for 24,539 genetic loci. For initial screening, we adopt the Pearson correlation coefficient proposed by \cite{pearson1895vii} to assess the association between the response variable and each genetic locus. The top 200 loci are selected and treated as linear covariates. Additionally, parity, maternal BMI, gestational age, and blood cotinine levels are  also included as linear covariates. Considering the potential nonliear relationship between maternal age and birth weight, maternal age is incorporated into the model as a nonlinear component.

Firstly, we assess the estimation performance of the FAPLM proposed in this paper, alongside the PLM from \cite{lian2019projected}, and the FALM from \cite{fanJ2023}. The evaluation is based on 200 random partitions of the dataset, with 50 subjects randomly selected as the training set and the remaining 15 as the test set. Using the model trained on each partition, we predict the response variables in the test set. Table \ref{Parameter estimation and prediction} reports  the average values of the following indicators to evaluate the estimation accuracy and predictive performance.
\begin{itemize}
    \item RMSE: $$\frac{1}{\sqrt{n_\text{test}}}\left\{\sum_{i=1}^{n_\text{test}}\left(\hat{Y}_i-Y_i^{*} \right)^2\right\}^{1/2}.$$
    Here $n_\text{test}=15$, $Y_i^{*}$ represents the actual response of the $i$-th observation in the test set, while $\hat{Y}_i$ denotes its corresponding predicted value. 
     \item Out-of-sample $R^2$: 
     $$R^2=1-\frac{\sum_{i=1}^{n_\text{test}}\left(\hat{Y}_i-Y_i^{*} \right)^2}{\sum_{i=1}^{n_\text{test}}\left(\bar{Y}_i-Y_i^{*} \right)^2},$$
     where $\bar{Y}_i=n_\text{test}^{-1}\sum_{i=1}^{n_\text{test}}Y_i^{*} $.
\end{itemize}
The results in Table \ref{Parameter estimation and prediction} show that FAPLM outperforms both PLM and FALM. When covariates exhibit significant correlations, applying PLM directly yields inaccurate estimators of parameters and suboptimal predictive performance, indicating its inability to capture the underlying data structure effectively. Additionally, FAPLM demonstrates superior predictive accuracy compared to FALM, highlighting the necessity of incorporating a partially linear structure to account for nonparametric effects.


\begin{table}[H]
					\small
					
					\renewcommand\arraystretch{1.1}
					\centering \tabcolsep 12pt \LTcapwidth 6in
					\caption{Prediction results of different methods based on 200 random partitions. The ``FAPLM'', ``PLM'' and   ``FALM'' represent the prediction results under FAPLM in this paper and PLM without incorporating the factor effect in \cite{lian2019projected}, and linear regression model accounting for latent factor effects in \cite{fanJ2023},  respectively.}
					\label{Parameter estimation and prediction}
					\begin{threeparttable}
						\begin{tabular}{cccc}
							\toprule
							 & FAPLM &PLM &FALM    \\ \midrule
							RMSE  &   280.94   &341.72 & 288.68    \\ 
							$R^2$&  0.75                    &  0.56 &0.68     \\ 
							\bottomrule
						\end{tabular}
					\end{threeparttable}
				\end{table}


We next illustrate the estimated function 
$\hat{g}(Z)$ using the FAPDT method proposed in this paper and the PDT method from \cite{zhu2019high}. 
The effect of mother’s age on infant birth weight is nonmonotone: the effect is first decreasing (up to age 25), then increasing (to about age 29), then decreasing (to about  age 34), then increasing (to about age 38), and decreasing again.
Furthermore, the nonlinear effects estimated by the two methods differ significantly. 
Specifically, the effect estimated using the FAPDT method varies markedly with maternal age, aligning more closely with real-world observations, while the effect estimated by the PDT method appears smoother and less variable.



\begin{figure}[H]%
					\centering
					\includegraphics[width=0.8\textwidth]{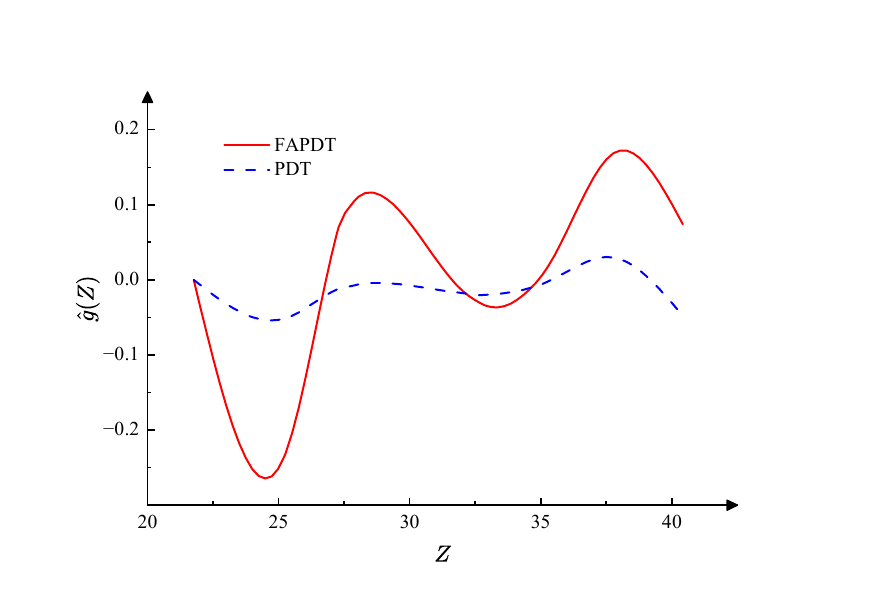}
					\caption{The estimated nonlinear effect $\hat{g}(Z)$ of mother’s age using two testing methods. The  ``FAPDT'' and ``PDT'' signify the results derived from the method in this paper   and  the approach in \cite{zhu2019high}, respectively.} \label{g_hat}
				\end{figure}

Finally, we utilize the proposed FAPDT, along with the PDT  from \cite{zhu2019high} and the FabT  from \cite{fanJ2023}, to assess the significance of the linear components. Critical values are obtained through 2000 bootstrap replications. For FabT, 10-fold cross-validation is employed to determine the tuning parameters, and iterated sure independent screening is used to estimate error variance.
Table \ref{$p$-values} reports the $p$-values for these tests. At a significance level of 0.05, the null hypothesis is rejected for all three methods, indicating that the linear components of the selected gene expression and corresponding clinical variables significantly influence infant birth weight.

\begin{table}[H]
					\small
					
					\renewcommand\arraystretch{1.1}
					\centering \tabcolsep 12pt \LTcapwidth 6in
					\caption{$p$-values  for FAPDT, PDT, and FabT: Results correspond to the method proposed in this paper and the approaches from \cite{zhu2019high} and \cite{fanJ2023}.}
					\label{$p$-values}
					\begin{threeparttable}
						\begin{tabular}{cccc}
							\toprule
							Test       & FAPDT   & PDT  & FabT  \\ \midrule
							$p$-value    & 0.027    & 0.016 & 0.046\\
							\bottomrule
						\end{tabular}
					\end{threeparttable}
				\end{table}

				\section{Conclusions and discussions}\label{Conclusions and discussions section}
Motivated by the challenges presented in the birth weight data, this paper introduces FAPLM to address the partially linear  dependence between the response variable and covariates, as well as the significant correlations among high-dimensional covariates. To achieve effective estimation and inference, we first employ the B-spline method to handle the nonparametric component and use the factor  model to capture the dependency structure among features. 
For parameter estimation, we apply the $\ell_1$-norm penalty and establish both $\ell_1$- and $\ell_2$-norm  error bounds, which agree with the minimax rate of the standard Lasso problems. Additionally, we develop a factor-adjusted projection debiased test to assess the significance of the linear component. To determine the critical value of the proposed test statistic FAPDT, we introduce the Gaussian multiplier bootstrap procedure. The theoretical validity of our methods is established under regularity conditions. Numerical results demonstrate the finite sample performance  of our proposed methodologies.

There are several promising directions for future research. First, our current work assume that the dimension of the nonparametric component is fixed. Exploring methods  to handle PLMs with high-dimensional linear and nonparametric components would be an intriguing avenue.  Second, our framework relies on the sub-Gaussian assumption for the error term. Developing inference techniques for high-dimensional PLMs with heavy-tailed errors is both crucial and of significant interest. We plan to address this challenging problem in the near future.

						\normalem
						\bibliographystyle{apalike}
						\bibliography{bibliography1020}
						
					\end{document}